\documentclass{article}
\usepackage{times}
\usepackage{latexsym}
\usepackage{verbatim}
\newtheorem{theorem}{Theorem}
\newtheorem{property}{Property}
\newtheorem{corollary}{Corollary}
\newtheorem{lemma}{Lemma}

\newtheorem{definition}{Definition}

\newcommand{\blackslug}{\mbox{\hskip 1pt \vrule width 4pt height 8pt
depth 1.5pt \hskip 1pt}}
\newcommand{\qed}{\quad\blackslug\lower 8.5pt\null\par\noindent}
\newenvironment{proof}{\par\noindent{\bf Proof:}}{\qed \par}

\newcommand{\cF}{\mbox{${\cal F}$}}

\newcommand{\cH}{\mbox{${\cal H}$}}

\newcommand{\cR}{\mbox{${\cal R}$}}

\newcommand{\eqdef}{\stackrel{\rm def}{=}}

\title{Foundations of non-commutative probability theory (Extended abstract)
\thanks{This work was partially supported
by the Jean and Helene Alfassa fund for
research in Artificial Intelligence}
}

\author{Daniel Lehmann\\School of Engineering
and\\Center for the Study of Rationality 
\\Hebrew University, Jerusalem 91904, Israel
}
\date{March 2009}

\begin{document}
\maketitle
\begin{abstract}
Kolmogorov's setting for probability theory is given an original generalization
to account for probabilities arising from Quantum Mechanics.
The sample space has a central role in this presentation
and random variables, i.e., observables, are defined in a natural way.
The mystery presented by the algebraic equations satisfied by
(non-commuting) observables that cannot be observed
in the same states is elucidated.
\end{abstract}

\section{Introduction} \label{sec:intro}
In Quantum Physics a state of a physical system defines random variables
corresponding to {\em observables} that are represented
by Hermitian operators.
These random variables cannot be treated
in the framework, laid down by Kolmogorov in the 30's, which is now
standard in probability theory.
The main reason is that, in the standard treatment,
real random variables are functions from the sample space to
the set of reals, implying that all points of the sample space {\em possess}
values for any random variable, whereas the standard understanding of Quantum
Physics requires that random variables that correspond to non-commuting
operators cannot both have a value at the same time.

This paper proposes a generalization of Kolmogorov's framework that
encompasses the non-commuting probabilities arising from Quantum Physics.
Contrary to previous efforts, known under the general term of Quantum Logic
and which~\cite{DallaChiara:01} surveys in an authoritative way,
in which the sample space is absent, this effort gives a central role to the
sample space.

\section{Kolmogorov's setting} \label{sec:K}
We shall recall the now classical setting laid down by Kolmogorov.
The description below is not the most succinct possible, but the reader
will have no problem showing it is equivalent to his/her favorite presentation.

We start with an arbitrary non-empty set $\Omega$,
the {\em sample space}, whose elements are called {\em points}.

\begin{definition} \label{def:sigma}
A set \mbox{$\cF \subseteq 2^{\Omega}$} of sets of sample points is
a $\sigma$-field iff it satisfies
\begin{enumerate}
\item \mbox{$\emptyset \in \cF$},
\item for any \mbox{$A \in \cF$}, the complement of $A$,
\mbox{$A^{c} = \Omega - A$} is a member of \cF,
\item for any finite or {\em countably infinite} sequence \mbox{$A_{i}$},
\mbox{$i \in I$} of {\em pairwise disjoint} elements
of \cF\ (for any \mbox{$i , j \neq i$}, \mbox{$A_{i} \cap A_{j} = \emptyset$})
their union \mbox{$\bigcup_{i \in I} A_{i}$} is a member of \cF.
\end{enumerate}
The elements of \cF\ are called {\em events}.
\end{definition}

\begin{definition} \label{def:prob}
\begin{sloppypar}
A probability measure (or, distribution) $p$ is a function
\mbox{$p : \cF \longrightarrow [0 , + \infty]$} such that:
\end{sloppypar}
\begin{enumerate}
\item \mbox{$p(\emptyset) = 0$},
\item \mbox{$p(\Omega) = 1$},
\item for any finite or {\em countably infinite} sequence \mbox{$A_{i}$},
\mbox{$i \in I$} of {\em pairwise disjoint} elements
of \cF, we have
\mbox{$p(\bigcup_{i \in I} A_{i}) =$}
\mbox{$\sum_{i \in I} p(A_{i})$}.
\end{enumerate}
\end{definition}

The definition of a random variable is now the following.
\begin{definition} \label{def:rand}
A random variable of a $\sigma$-field \mbox{$\langle \Omega , \cF \rangle$}
to a $\sigma$-field \mbox{$\langle \cR , \Sigma \rangle$}
is a {\em measurable} function \mbox{$X : \Omega \longrightarrow \cR$},
that is to say a function $X$ such that the inverse image by $X$
of any event of $\Sigma$ is an event of \cF.
\end{definition}

\section{Extant work} \label{sec:extant}
The foundation of Quantum Logic was laid by Birkhoff and von Neumann
in~\cite{BirkvonNeu:36} which set the frame for later work in Quantum Logic.
This frame is based on the classical views that quantic propositions are
either true or false, that propositions can be composed using negation,
conjunction and disjunction,  and that the structure to be studied is the
consequence relation: which propositions follow from other propositions or sets
of propositions.
The algebraic structure of such propositions is naturally seen to be an ordered
structure, in fact a lattice. Birkhoff and von Neumann noticed that the
lattice of interest is not, in general, distributive. Quantum Logic therefore
studied non-distributive complemented lattices, satisfying a property weaker
than distributivity: modularity was advocated by~\cite{BirkvonNeu:36} but
most researchers opted for the even weaker orthomodularity.

The probabilistic aspect of Quantum Physics is probably its most
revolutionary feature. There is no doubt that a physicist will consider
the fact that, in Quantum Physics, a state can define, even in principle,
only the probability of observations as more immediately revolutionary than
the fact that disjunction does not distribute over conjunction. We shall
now describe the way Quantum Logic deals with probabilities. Its analysis
of classical probabilities relies on the observation that a $\sigma$-field
defines a Boolean algebra with countable l.u.b's.
A (classical) probability measure is therefore a function that
attaches a real number (its probability) to every element of a Boolean algebra
and satisfies certain conditions. The concrete algebra of subsets presented
in Kolmogorov's setting is replaced by an abstract Boolean algebra.
By Stone's representation theorem, there is no loss here since any Boolean
algebra is isomorphic to a concrete algebra of subsets.
Probability measures in Quantum Logic are therefore analyzed as functions
assigning a probability to every element of an orthomodular lattice that
satisfy certain properties.
But orthomodular (or modular) lattices are not, in general,
lattices of sets and the sample space disappears from the picture.
This has three serious drawbacks.
First the intuitive idea that probability of an event is, in some sense,
the measure of the ``size'' of a set of possibilities cannot be carried on.
Secondly, the definition of a random variable, which requires a sample space,
is not possible. Thirdly, the special case of classical probabilities is
characterized by the boolean character of the lattice and this may seem at odds
with the view generally held by physicists that classical physics is the
special case of quantum physics in which all operators commute: it is
difficult to see boolean lattices as {\em commutative orthomodular lattices}.
A family of algebras generalizing boolean algebras has been proposed
in~\cite{Lehmann_andthen:JLC}
and boolean algebras are exactly the commutative algebras of the family.
The relation of those algebras to the present work needs further study.

The first concern has been addressed by setting additional requirements,
concerned with Atomicity and Covering,
on the lattice structure: see for example axioms
H1 and H2 in~\cite{Pitowsky:TheoryofProb}.
Such properties are {\em not} satisfied in Boolean algebras and therefore
classical probabilities are {\em not} a special case of Quantum probabilities.
Random variables may then be defined by functions on the atoms
of the structure.

This work proposes a framework for probability theory that generalizes
Kolmogorov's and that encompasses Quantum Probability. Classical probability
is a special case of Quantum Probability. The sample space is not eliminated,
but given some additional structure: it is an Similarity-Projection (SP)
structure.
These have been defined and studied in~\cite{SP:IJTP}. They abstract
from the real scalar product in Hilbert spaces.

\section{A more general setting} \label{sec:general}
We shall generalize Kolmogorov's setting by assuming some structure
on the sample space $\Omega$. We assume there is a {\em similarity} function
\mbox{$s : \Omega \times \Omega \longrightarrow \cR$} that associates
a real number, their similarity, to any two sample points.
Think of $x$ and $y$ as unitary vectors in a Hilbert space and think
of \mbox{$s(x, y)$} as their real scalar product squared:
\mbox{$\mid \langle x , y \rangle \mid^{2}$}.
We shall assume that the pair \mbox{$\langle \Omega , s \rangle$}
is a Similarity-Projection (SP) structure as defined in~\cite{SP:IJTP},
where $p$ was used instead of $s$. Intuitively, SP-structures may be understood
as one-dimensional subspaces of a Hilbert space with holes.
A set of $n$ elements is an n-dimensional Hilbert space with very big holes.
We shall now recall the properties of SP-structures that we shall need,
with the necessary definitions and notations.
We restrict our attention to {\em standard} SP-structures as defined
in~\cite{SP:IJTP}. The definition of a standard SP-structure is recalled
in Appendix~\ref{app:SP}.

The properties below are the ones we shall use in the sequel, they should not
be taken as a definition of SP-structures. A physically and epistemologically
motivated definition of SP-structures may be found in~\cite{SP:IJTP} where
the properties below are proved out of a set of seemingly weak assumptions.
Property~\ref{cardinality} that is so striking in Hilbert spaces is not an
assumption, it follows from more basic properties. Similarly for
Property~\ref{arb_inter}. Property~\ref{cont} seems original.
It means that the similarity function \mbox{$s(x , y)$} is, in a sense,
continuous: for $\epsilon > 0$, close enough to $0$,
if \mbox{$s(x , y) \geq 1 - \epsilon$}, then for any
\mbox{$z \in \Omega$} the difference
\mbox{$s(x , z) - s(y , z)$} is of order $\sqrt{\epsilon}$.

In the following \mbox{$x , y , z$} are arbitrary elements of the sample space
$\Omega$ and \mbox{$A , B$} are arbitrary subsets of $\Omega$.
\begin{enumerate}
\item \label{eq} \mbox{$s(x , y) \in [0 , 1]$},
and \mbox{$x = y$} iff \mbox{$s(x , y) = 1$},
\item \mbox{$s(y , x) = s(x , y)$},
\item $x$ and $y$ are said to be {\em orthogonal},
written \mbox{$x \perp y$} iff \mbox{$s(x , y) = 0$}, we say that
$x$ is orthogonal to $A$ and write \mbox{$x \perp A$} iff \mbox{$x \perp y$}
for every \mbox{$y \in A$}, we say that $A$ and $B$ are orthogonal
and write \mbox{$A \perp B$}
iff \mbox{$z \perp B$} for every \mbox{$z \in A$},
\item $A$ is said to be an {\em ortho-set} iff
all pairs of distinct elements of $A$ are orthogonal,
\item for any ortho-set $A$,
\mbox{$s(x , A) \eqdef \sum_{y \in A} s(x , y) \leq 1$},
\item $B$ is said to be a {\em subspace} and $A$ is said to be a basis for $B$
iff $A$ is an ortho-set and
\mbox{$B = \{x \in \Omega \mid s(x , A) = 1 \}$},
\item \label{cardinality}
if $B$ is a subspace all bases for $B$ have the same cardinality,
\item \label{arb_inter} if \mbox{$A_{i}$} for \mbox{$i \in I$} are subspaces,
then their intersection
\mbox{$\bigcap_{i \in I} A_{i}$} is also a subspace: subspaces are closed under
arbitrary intersections,
\item $\emptyset$ is a subspace, $\Omega$ is a subspace,
\item the orthogonal complement of any subset $A$ is defined by:
\[
A^{\perp} \eqdef \{x \in \Omega \mid x \perp A\},
\]
\item $A^{\perp}$ is a subspace, if $A$ is a subspace
then \mbox{$(A^{\perp})^{\perp} = A$},
\mbox{$\emptyset^{\perp} = \Omega$}, \mbox{$\Omega^{\perp} = \emptyset$},
\item \label{proj} for any subspace $A$ and any \mbox{$x \in \Omega$},
if $x$ is not orthogonal to $A$, there is a unique
\mbox{$t(x , A) \in A$} such that \mbox{$s(x, t(x, A)) = s(x , A)$}
and for every \mbox{$y \in A$} one has
\mbox{$s(x , y) =$} \mbox{$s(x , t(x , A)) \, s(t(x , A) , y)$},
\item \label{cont}
\begin{equation} \label{eq:cont}
s(z , x) \leq
\end{equation}
\[
s(z , y) \: + \: 1 / 2 \, \sqrt{1 - s(x , y)}
\: + \: (1 - s(x , y)).
\]
\end{enumerate}

Note that the seemingly natural {\em triangular inequality}:
\mbox{$s(x , y) \leq s(x , z) \, s(z , y)$} is not a property of
SP-structures. It does not hold in Hilbert spaces.
A classical SP-structure is defined to be
a structure in which \mbox{$s(x , y) = 0$}
whenever \mbox{$x \neq y$}. In a classical SP-structure $x$ and $y$
are orthogonal iff they are different, and $A$ and $B$ are orthogonal
iff they are disjoint. Any set $A$ is a subspace. The orthogonal complement
of a set $A$ is its set complement: \mbox{$\Omega - A$}.

\section{Properties of SP structures} \label{sec:SP}
We present here properties of SP-structures that have not been presented
in~\cite{SP:IJTP}.
We define the sum $A \oplus B$ of any two {\em subsets} of $\Omega$.
The set $A \oplus B$ is the minimal {\em subspace} that contains $A$ and $B$.

\begin{definition} \label{def:*sum}
Let \mbox{$\langle \Omega , s \rangle$} be an SP-structure.
If \mbox{$A , B \subseteq \Omega$}, their sum
\mbox{$A \oplus B$} is defined to be the smallest subspace
including \mbox{$A \cup B$}:
\[
A \oplus B \: = \:
\bigcap_{X {\rm \ is \ a \ subspace}, \ A \cup B \subseteq X} X.
\]
This definition is correct since, as noticed in~\ref{arb_inter} above,
subspaces are closed under intersection.
One easily sees that sum is commutative, associative and monotone:
\mbox{$A \subseteq A'$} implies \mbox{$A \oplus B \subseteq A' \oplus B$}.
Therefore the sum of any family (finite or infinite) of subsets
is well-defined:
\mbox{$\bigoplus_{i \in I} A_{i}$} is the intersection of all subspaces
including \mbox{$\cup_{i \in I} A_{i}$}.
\end{definition}
In a classical structure sum is union: \mbox{$A \oplus B = A \cup B$}.

\begin{lemma} \label{le:sum}
For any subspaces $A$, $B$: \mbox{$A \oplus A^{\perp} = \Omega$} and
\mbox{$A \cap A^{\perp} = \emptyset$}.
\end{lemma}
\begin{proof}
Let \mbox{$x \in \Omega$}. Let $B$ be a basis for $A$.
By Theorem~1 of~\cite{SP:IJTP} there is a basis for $\Omega$ that includes
$B$. Let \mbox{$B \cup B'$} be this basis: \mbox{$p(x , B) + p(x , B') = 1$}
and \mbox{$B' \subseteq A^{\perp}$}.
Any subspace that includes $B$ and $B'$ must be $\Omega$.
We have shown that \mbox{$A \oplus A^{\perp} = \Omega$}.

If \mbox{$x \in A \cap A_{\perp}$} we must have \mbox{$s(x , x) = 0$},
contradicting property~\ref{eq} above. We have shown that
\mbox{$A \cap B = \emptyset$}.
\end{proof}

We shall now show that the set of subspaces of an SP-structure is
an orthomodular complemented lattice. The lattice of closed subspaces of
a Hilbert space shows that it is not always modular. The structure we are
interested in is an orthomodular lattice, but note that we have additional
structure given by the similarity function.
\begin{theorem} \label{the:ortho}
Let \mbox{$\langle \Omega , s \rangle$} be an SP-structure.
The set of subspaces of $\Omega$ is a complete complemented orthomodular
lattice, if one takes \mbox{$A \leq B$} iff \mbox{$A \subseteq B$}.
Least upper bound is $\oplus$ and greatest lower bound is intersection.
\end{theorem}
\begin{proof}
The relation $\leq$ is obviously a partial order and, since,
by~\ref{arb_inter}, subspaces are closed under intersections, intersections
are greatest lower bounds. By definition sums are least upper bounds
and the lattice is complete. Lemma~\ref{le:sum} shows that it is
a complemented lattice. Orthomodularity is a consequence of Theorem~8
of~\cite{SP:IJTP}: if \mbox{$A \subseteq C$} are subspaces any basis $B$ for
$A$ can be extended into a basis $B \cup B'$ for $C$ and therefore
\mbox{$C \subseteq$}
\mbox{$B \oplus B' \subseteq$} \mbox{$A \oplus A^{\perp} \cap C$}.
\end{proof}

De Morgan's laws hold in any orthocomplemented lattice.
\begin{corollary} \label{le:deMorgan}
\begin{sloppypar}
For any subspaces $A$ and $B$
\mbox{$(A \cap B)^{\perp} =$} \mbox{$A^{\perp} \oplus B^{\perp}$} and
\mbox{$(A \oplus B)^{\perp} =$} \mbox{$A^{\perp} \cap B^{\perp}$}.
These equalities extend to arbitrary infinite sums and intersections.
\end{sloppypar}
\end{corollary}

We shall now generalize the similarity $s$ to {\em arbitrary subspaces}
of $\Omega$.
\begin{definition} \label{def:sAB}
Let \mbox{$A , B \subseteq \Omega$} be subspaces. We wish to define a measure
of their similarity, denoted \mbox{$s(A , B)$}.
Let \mbox{$x \in \Omega$}, we shall define \mbox{$\tau(x, A, B)$}
to be the similarity of $A$ and $B$ from the vantage point $x$.
Then we let \mbox{$s(A , B) =$}
\mbox{$\liminf\{ \tau(x, A, B) \mid x \in \Omega \}$}.
Now let us define \mbox{$\tau(x, A, B)$}.
In case \mbox{$x \not \perp A$} and \mbox{$x \not \perp B$}, let
\mbox{$\tau(x, A, B) =$}
\mbox{$s(t(x, A) , t(x, B))$}.
If \mbox{$x \perp A$}, we let \mbox{$\tau(x, A, B) =$}
\mbox{$1 - s(x, B)$}.
If \mbox{$x \perp B$}, we let \mbox{$\tau(x, A, B) =$}
\mbox{$1 - s(x, A)$}.
\end{definition}
Note that if \mbox{$x \perp A$} and \mbox{$x \perp B$} both last conditions
give \mbox{$\tau(x, A, B) = 1$}.
Note also that \mbox{$s(A , B) =$}
\mbox{$s(B , A)$}.

\begin{theorem} \label{the:sAB_sing}
For any \mbox{$x , y \in \Omega$},
we have \mbox{$s(\{x\} , \{y\}) =$}
\mbox{$s(x , y)$}.
\end{theorem}
\begin{proof}
If \mbox{$z \not \perp x$} and \mbox{$z \not \perp y$} we have
\mbox{$\tau(z , \{x\} , \{y\}) =$} \mbox{$s(x , y)$}.
If \mbox{$z \perp x$}, we have
\mbox{$\tau(z , x , y) =$} \mbox{$1 - s(z , y) \geq$} \mbox{$s(x , y)$}.
If \mbox{$z \perp y$}, we have
\mbox{$\tau(z , x , y) =$} \mbox{$1 - s(z , x) \geq$} \mbox{$s(x , y)$}.
\end{proof}

\begin{theorem} \label{the:sAB_le}
If \mbox{$x \in A$}, then \mbox{$s(A , B) \leq$}
\mbox{$s(x , B)$}.
\end{theorem}
\begin{proof}
\begin{sloppypar}
Suppose \mbox{$x \not \perp B$}.
Then one has \mbox{$s(A , B) \leq$}
\mbox{$\tau(x , A , B) =$}
\mbox{$s(x , t(x , B)) =$} \mbox{$s(x , B)$}.
If \mbox{$x \perp B$} we have
\mbox{$s(A , B) \leq$}
\mbox{$\tau(x , A , B) =$}
\mbox{$1 - s(x , A) =$} \mbox{$0 \leq$} \mbox{$s(x , B)$}.
\end{sloppypar}
\end{proof}
In general, \mbox{$s(\{x\} , B) <$} \mbox{$s(x , B)$}.

\begin{theorem} \label{the:sAB}
\mbox{$s(A , B) = 1$} iff \mbox{$A = B$}.
\end{theorem}
\begin{proof}
Let \mbox{$A = B$}. If \mbox{$x \not \perp A$} we have
\mbox{$\tau(x , A , B) =$} \mbox{$s(t(x , A), t(x , A)) = 1$}.
If \mbox{$x \perp A$} we have
\mbox{$\tau(x , A , B) =$} \mbox{$1 - s(x, B) =$} \mbox{$1$}.
Assume, now, that \mbox{$s(A , B) = 1$}.
For every \mbox{$x \in \Omega$} we have \mbox{$\tau(x , A , B) = 1$}.
Suppose \mbox{$x \in A$}.
If \mbox{$x \not \perp B$} we have \mbox{$\tau(x , A , B) = $}
\mbox{$s(x, t(x, B)) = 1$} and \mbox{$x \in B$}.
If \mbox{$x \perp B$}, we have
\mbox{$1 - s(x , t(x , A)) = 1$} and
\mbox{$s(x, x) = 0$}, a contradiction.
We conclude that \mbox{$A \subseteq B$}.
Similarly we have \mbox{$B \subseteq A$} and we conclude that
\mbox{$A = B$}.
\end{proof}

We may now generalize Inequality~\ref{eq:cont}.
\begin{theorem} \label{the:sAB_in}
Let \mbox{$A , B , C \subseteq \Omega$} be subspaces. We have
\[
s(A , B) \leq s(A , C) + 1 / 2 \sqrt{1 - s(B , C)} + 1 - s(B , C).
\]
\end{theorem}
\begin{proof}
It is enough to show that, for any \mbox{$x \in \Omega$}, we have
\begin{equation} \label{eq:prove}
\tau(x , A , B) \leq
\end{equation}
\[
\tau(x , A , C) +
1 / 2 \sqrt{1 - s(B , C)} + 1 - s(B , C).
\]
Suppose, first, that \mbox{$x \not \perp A$}, \mbox{$x \not \perp B$}
and \mbox{$x \not \perp C$}.
Let \mbox{$w = t(x , A)$}, \mbox{$y = t(x , B)$} and \mbox{$z = t(x , C)$}.
By Inequality~\ref{eq:cont} we have:
\[
s(w , y) \leq s(w , z) + 1 / 2 \sqrt{1 - s(y , z)} + 1 - s(y , z)
\]
and therefore
\[
\tau(x , A , B) \leq
\]
\[
\tau(x , A , C) +
1 / 2 \sqrt{1 - \tau(x , B , C)} + 1 - \tau(x , B , C).
\]
But \mbox{$s(B , C) \leq$} \mbox{$\tau(x , B , C)$} and Equation~\ref{eq:prove}
is proved.
Suppose, now, that \mbox{$x \not \perp A$}, \mbox{$x \not \perp B$} but
\mbox{$x \perp C$}. Let \mbox{$w = t(x , A)$} and \mbox{$y = t(x , B)$}.
We must show that
\[
s(w , y) \leq 1 - s(x , A) + 1 / 2 \sqrt{1 - s(B , C)} + 1 - s(B , C).
\]
But \mbox{$s(x , A) =$} \mbox{$s(x , w)$}.
We know that \mbox{$x \perp C$} and \mbox{$s(w , x) + s(w , C) \leq 1$}.
It is therefore enough to show that
\[
s(w , y) \leq s(w , C) + 1 / 2 \sqrt{1 - s(B , C)} + 1 - s(B , C).
\]
If \mbox{$w \not \perp C$}, by Inequality~\ref{eq:cont}, we have
\[
s(w , y) \leq
\]
\[
s(w , t(w , C)) + 1 / 2 \sqrt{1 - s(\{y\} , C)} +
1 - s(\{y\} , C).
\]
We conclude by Theorem~\ref{the:sAB_le}. If \mbox{$w \perp C$},
it is enough to show that
\mbox{$s(w , y) + s(B , C) \leq 1$}. But \mbox{$s(B , C) \leq$}
\mbox{$s(y , C)$} by Theorem~\ref{the:sAB_le} and
we have \mbox{$s(y , w) + s(y , C) \leq 1$}.

The case \mbox{$x \not \perp A$}, \mbox{$x \not \perp C$} but
\mbox{$x \perp B$} is treated similarly.

If \mbox{$x \not \perp A$}, \mbox{$x \perp B$} and \mbox{$x \perp C$}
we must show that
\[
1 - s(x , A) \leq 1 - s(x , A) + 1 / 2 \sqrt{1 - s(B , C)} + 1 - s(B , C)
\]
which is obvious.

We are left with the case \mbox{$x \perp A$}.
We must show that
\[
1 - s(x , B) \leq 1 - s(x , C) +
1 / 2 \sqrt{1 - s(B , C)} + 1 - s(B , C),
\]
or equivalently
\[
s(x , C) \leq s(x , B) +
1 / 2 \sqrt{1 - s(B , C)} + 1 - s(B , C).
\]
If \mbox{$x \perp C$} the claim is obvious.
Assume \mbox{$x \not \perp C$} and let \mbox{$z = t(x , C)$}.
If \mbox{$x \perp B$} we have \mbox{$s(z , x) + s(z , B) \leq 1$}
and therefore
\mbox{$s(x , C) \leq 1 - s(z , B) \leq 1 - s(B , C)$}.
If, last, \mbox{$x \not \perp B$} and we let \mbox{$y = t(x , B)$},
by Inequality~\ref{eq:cont} we have:
\[
s(x , z) \leq s(x , y) + 1 / 2 \sqrt{1 - s(y , z)} + 1 - s(y , z)
\]
and therefore
\[
s(x , C) \leq s(x , B) +
1 / 2 \sqrt{1 - \tau(x , B , C)} + 1 - \tau(x , B , C) \leq
\]
\[
s(x , B) +
1 / 2 \sqrt{1 - s(B , C)} + 1 - s(B , C).
\]
\end{proof}

In classical SP-structures $s(A , B)$ is equal to $1$ iff \mbox{$A = B$}
and equal to $0$ otherwise.

\section{$\sigma$*-fields} \label{sec:sigma*}
We want to generalize Definition~\ref{def:sigma}, i.e.,
the definition of a $\sigma$-field on a set
$\Omega$ to that of a {\em $\sigma$*-field} over an SP-structure
\mbox{$\langle \Omega , s \rangle$}. As expected, we require that
the $\sigma$*-field \cF\ be closed under countably many sums
(sums generalize unions) and under orthogonal complement
(they generalize set-complements).
But we require the elements of a $\sigma$*-field \cF\ are {\em subspaces},
not arbitrary subsets, of $\Omega$.

\begin{definition} \label{def:*-field}
Let \mbox{$\langle \Omega , s \rangle$} be an SP-structure.
A set \cF\ of {\em subspaces} of \mbox{$\langle \Omega , s \rangle$}is said
to be a {\em $\sigma$*-field} over \mbox{$\langle \Omega , s \rangle$} iff:
\begin{enumerate}
\item \label{empty} \mbox{$\emptyset \in \cF$},
\item \label{perp}
for every \mbox{$A \in \cF$}, its orthogonal complement $A^{\perp}$ is in
\cF,
\item \label{den_sum}
for any set, finite or {\em countably} infinite $A_{i}$, \mbox{$i \in I$},
of pairwise orthogonal elements of \cF, its sum
\mbox{$\bigoplus_{i \in I} A_{i}$} is in \cF.
\end{enumerate}
Elements of \cF\ are called {\em events}.
\end{definition}
Note that \mbox{$\Omega = \emptyset^{\perp}$} is an event.

If \mbox{$\langle \Omega , s \rangle$} is a classical SP-structure then
the notion of a $\sigma$*-field on the structure is equivalent to
that of a $\sigma$-field on $\Omega$.

\begin{lemma} \label{le:inter}
Assume \cF\ is a $\sigma$*-field on \mbox{$\langle \Omega , s \rangle$}, $I$
is finite or {\em countably} infinite, and for any \mbox{$i \in I$},
$A_{i}$ is an element of \cF.
Then, the intersection
\mbox{$\bigcap_{i \in I} A_{i}$} is in \cF.
\end{lemma}
\begin{proof}
By Corollary~\ref{le:deMorgan}.
\end{proof}

\begin{corollary} \label{le:Fortho}
Any $\sigma$*-field is a bounded complemented orthomodular
lattice, if one takes \mbox{$A \leq B$} iff \mbox{$A \subseteq B$}.
Least upper bound is $\oplus$ and greatest lower bound is intersection.
Countably infinite sets have l.u.b. and g.l.b. but the lattice is not,
in general, complete.
\end{corollary}

\section{Probability distributions} \label{sec:prob_dist}
We may now generalize Definition~\ref{def:prob}.
We shall define *-probabilities that attach a probability to
events of a $\sigma$*-field. Note that {\em states} in Quantum Physics
are such *-probabilities.
Our first three conditions are those of Definition~\ref{def:prob},
but a fourth condition is added to ensure that probabilities are,
in a sense, {\em continuous}.
If the subspaces $A$ and $B$ are {\em close}, i.e.,
\mbox{$s(A , B)$} is close to $1$,
then we expect $p(A)$ and $p(B)$ to be close to each other.

\begin{definition} \label{def:*-prop}
Assume \mbox{$\langle \Omega , s \rangle$} is an SP-structure, and
\cF\ is a $\sigma$-* field on \mbox{$\langle \Omega , s \rangle$}.
A *-probability on \mbox{$\langle \Omega , s, \cF \rangle$}
is a function \mbox{$p : \cF \longrightarrow [0 , + \infty]$} that satisfies:
\begin{enumerate}
\item \mbox{$p(\emptyset) = 0$},
\item \mbox{$p(\Omega) = 1$},
\item for any finite or countably infinite set, \mbox{$A_{i}$},
\mbox{$i \in I$} of pairwise orthogonal elements of \cF\ one has:
\mbox{$p(\bigoplus_{i \in I} A_{i}) = \sum_{i \in I} p(A_{i})$},
\item for any events $A$, $B$, we have
\begin{equation}
p(A) \leq
\end{equation}
\[
p(B) + 1 / 2 \, \sqrt{1 - s(A , B)} + (1 - s(A , B)).
\]
\end{enumerate}
\end{definition}
Note that, in our third condition, the sum
\mbox{$\bigoplus_{i \in I} A_{i}$} is an event by Definition~\ref{def:*-field}.
Our fourth condition is taken from~\ref{cont} above, which has been shown
to be tight in~\cite{SP:IJTP}.

It is clear that convex combinations of probabilities are probabilities.
\begin{lemma} \label{le:mixed}
Assume \mbox{$\langle \Omega , s \rangle$} and \cF\ are fixed.
If for any \mbox{$i \in I$} \mbox{$p_{i}$} is a probability
and \mbox{$w_{i} \in [0 , 1]$}
are such that \mbox{$\sum_{i \in I} w_{i} = 1$}, then
\mbox{$q = \sum_{i \in I} w_{i} \, p_{i}$} defined by
\mbox{$q(A) = \sum_{i \in I} w_{i} \, p_{i}(A)$} for any \mbox{$A \in \cF$}
is a probability.
\end{lemma}

\section{Pure and mixed states} \label{sec:pureandmixed}
In Quantum Physics pure states have a dual aspect: they are points of the
sample space, i.e., elements of the subspaces
representing quantic propositions, but they also attach probabilities to
points and subspaces (the {\em transition probability}).
This simply generalizes the fact that a point $x$ in
the sample space can be identified, in Kolmogorov's setting,
with the probability distribution that gives probability one to all events that
contain $x$ and probability zero to all other events. Probabilities attached
to points in the sample space are called {\em pure states} in Quantum Physics.

\begin{theorem} \label{the:px}
Assume \mbox{$\langle \Omega , s \rangle$} is an SP-structure, and
\cF\ is a $\sigma$-* field on \mbox{$\langle \Omega , s \rangle$}.
Let \mbox{$x \in \Omega$} be a point in the sample space.
One may define a *-probability $p_{x}$ by:
\mbox{$p_{x}(B) = s(x, B) = \sum_{y \in A} s(x , y)$} for any event $B$ and
any basis $A$ for $B$. Such probabilities are called {\em pure states} and
the set of pure states will be denoted by $P(\Omega)$.
Convex combinations of pure states are called mixed states. The set
of mixed states will be denoted \mbox{$M(\Omega)$}. We shall represent
mixed states as convex combinations of points of the sample space:
\mbox{$p = \sum_{i \in I} r_{i} \, x_{i}$} for non-negative real numbers
$r_{i}$ such that \mbox{$\sum_{i \in I} r_{i} = 1$} and
\mbox{$x_{i} \in \Omega$} for \mbox{$i \in I$}.
\end{theorem}
\begin{proof}
Obviously \mbox{$p_{x}(\emptyset) = 0$} and \mbox{$p_{x}(\Omega) = 1$}.
Suppose now that \mbox{$B_{i}$}, \mbox{$i \in I$} is a family of
pairwise orthogonal events. We have \mbox{$s(x, \bigoplus_{i \in I} B_{i}) =$}
\mbox{$\sum_{i \in I} s(x, B_{i})$} since a basis for the sum is the union of
bases for the B's.
To check the last (continuity) property of probability measures,
assume, first, that $x$ is orthogonal to neither $A$ nor $B$.
By properties~\ref{proj} and~\ref{cont} of Section~\ref{sec:general}
and by Definition~\ref{def:sAB} we have:
\[
s(x, A) = s(x , t(x , A)) \leq
\]
\[
s(x , t(x, B)) + 1 /2 \sqrt{1 - s(t(x, A), t(x, B))} +
\]
\[
(1 - s(t(x, A), t(x, B))) \leq
\]
\[
s(x, B) + 1 / 2 \sqrt{1 - s(A, B)} + (1 - s(A, B)).
\]
If \mbox{$x \perp A$}, the claim is obvious.
Suppose, now that \mbox{$x \perp B$} and \mbox{$x \not \perp A$}.
We have \mbox{$s(A, B) \leq$} \mbox{$\tau(x , A , B) =$} \mbox{$1 - s(x, A)$}.
Therefore \mbox{$s(x, A) \leq 1 - s(A, B)$}.
\end{proof}

Gleason's theorem~\cite{Gleason:57} says that,
for any SP-structure defined by the rays
of a Hilbert space of dimension larger than 2,  any probability measure
on the $\sigma$*-field of all closed subspaces is a mixed state.
Notice that the result does not hold for Hilbert spaces of dimension $2$.
For physical systems of dimension $2$ there are probabilities that are not
mixed states. Nevertheless it seems that the only probabilities found useful
to study such systems in quantum physics are mixed states.
The reason may be hidden
in the preparation of quantic systems: one seems to know how to prepare
a system in any mixed state but not in any state corresponding to a probability
measure that is not mixed. Therefore one is probably justified in restricting
one's attention to mixed states.
\begin{sloppypar}
A most important remark is that the set $M(\Omega)$ of all convex
combination of pure states is not a free structure. We may well have,
for example \mbox{$1 / 2 \, p_{x} \: + \: 1 / 2 \, p_{y} =$}
\mbox{$1 / 2 \, p_{w} \: + \: 1 /2 \, p_{z}$} with
\mbox{$x \neq w$} and \mbox{$x \neq z$}.
A topic for further study is the characterization of those transformations
\mbox{$\tau : P(\Omega) \longrightarrow M(\Omega)$} for which
\mbox{$\sum_{i \in I} r_{i} p_{x_{i}} =$}
\mbox{$\sum_{j \in J} s_{j} p_{y_{j}}$} implies
\mbox{$\sum_{i \in I} r_{i} \tau(p_{x_{i}}) =$}
\mbox{$\sum_{j \in J} s_{j} \tau(p_{y_{j}})$}.
\end{sloppypar}
In classical structures, mixed states are discrete probability measures
and therefore the remainder of this paper generalizes only discrete
probability theory. A generalization of continuous probability
theory is probably necessary to understand systems with observables
that can take a continuum of values.

\section{Random variables} \label{sec:random_var}
The definition of *-random variables, generalizing Definition~\ref{def:rand}
requires some thinking.

\begin{definition} \label{def:*rand}
Let \mbox{$\langle \Omega_{i} , s_{i} \rangle$}
be SP-structures, and $\cF_{i}$ be $\sigma$-* fields on
\mbox{$\langle \Omega_{i} , s_{i} \rangle$} for \mbox{$i = 1 , 2$}.
We want a random variable to give values in $\Omega_{2}$ to elements of
$\Omega_{1}$. So it seems a random variable $X$ should be a function
\mbox{$\Omega_{1} \longrightarrow \Omega_{2}$}.
But we noticed in Section~\ref{sec:intro} that non-commuting observables cannot
be defined have values at the same sample points. Therefore we must accept
the idea that $X$ be a {\em partial} function
\mbox{$X : \Omega_{1} \longrightarrow \Omega_{2}$}.
In the classical case of Definition~\ref{def:rand}, the function is a
{\em total} function and therefore we shall require that $X$ be defined on some
basis for $\Omega_{1}$. In the classical case $\Omega_{1}$ is the only basis
and therefore $X$ must be total.
We, then, as usual, require that the inverse image by $X$ of any element
of $\cF_{2}$ be an element of $\cF_{1}$.
Guided by the fact that, in the classical case, if $A$, $B$ are {\em disjoint}
elements of $\Omega_{2}$, their inverse images \mbox{$X^{-1}(A)$}
and \mbox{$X^{-1}(B)$} are disjoint, we require that if
\mbox{$A , B \in \cF_{2}$} and \mbox{$A \perp B$}, we have
\mbox{$X^{-1}(A) \perp X^{-1}(B)$}.
\end{definition}

Real random variables are important enough to justify a specialization
of Definition~\ref{def:realrand}
\begin{definition} \label{def:realrand}
Let \mbox{$\langle \Omega ,  s \rangle$}
be an SP-structure, and \cF\ a $\sigma$-* field on
\mbox{$\langle \Omega , s \rangle$}.
A real random variable $X$ is a partial function
\mbox{$X : \Omega \longrightarrow \cR$} that is defined on some
basis for $\Omega$ and such that the inverse image by $X$ of any
Lebesgue-measurable subset of \cR\ is an element of $\cF$ and such that
the inverse images of any two disjoint such subsets are orthogonal
elements of $\cF$.
\end{definition}

Note that Definition~\ref{def:realrand} ensures that
the set of points of the sample space $\Omega$ on which a random variable $X$
is defined is a set of
pairwise orthogonal subspaces (generalizing eigensubspaces)
whose sum is $\Omega$.

A real random variable is a partial function, but it defines a total function:
its expected value in each state. There is no problem in
considering that expected values of non-commuting observables are both defined
at the same time. This total function can be even defined on mixed states.

\begin{definition} \label{def:expect}
Let $X$ be a real random variable as above and suppose it takes only a
countable set of values: $r_{i}$ for \mbox{$i \in I$}.
Let \mbox{$p \in M(\Omega)$} be any mixed state.
We define $\hat{X}(p)$ as
\mbox{$\sum_{i \in I} r_{i} \, p(X^{-1}(r_{i}))$}.
\end{definition}

\begin{theorem} \label{the:expect}
Let $X$ be a random variable as in Definition~\ref{def:expect}.
Let \mbox{$x \in \Omega$} and assume \mbox{$B =$}
\mbox{$\{ b_{i} \mid i \in I\}$} is a basis for $\Omega$ on which $X$ is
defined. Then \mbox{$\hat{X}(p_{x}) =$}
\mbox{$\sum_{i \in I} X(i) \, s(x , b_{i})$}.
\end{theorem}
\begin{proof}
For any \mbox{$a \in \cR$}, let \mbox{$J(a) \subseteq I$} be the
set of indexes $i$ for which \mbox{$X(i) = a$}.
The subspace \mbox{$\oplus_{i \in J} b_{i}$} spanned by the corresponding
basis elements is equal to \mbox{$X^{-1}(a)$}, and
\mbox{$p_{x}(X^{-1}(a)) =$} \mbox{$s(x , \oplus_{i \in J} b_{i}) =$}
\mbox{$\sum_{i \in J} s(x , b_{i})$}.
\end{proof}

\section{Future work} \label{sec:future}
In Quantum Physics, operators, and particularly
self-adjoint operators, play a central role.
Operators can be composed and their commutation properties represent
important physical information. One should try to reflect this
{\em transformational} aspect into our present framework, in terms of
properties of *-random variables. We hope to be able to characterize
classical Kolmogorov's probability theory as the special case
of *-probabilities in which random variables commute.

\section*{Acknowledgements}
I want to thank Itamar Pitowsky: a lecture of his started me thinking on
the relation of my previous work to Quantum Probability
and his help and advice are gratefully acknowledged and Benjamin Weiss for
his help.

\appendix
\section{SP-structures} \label{app:SP}
We recall here the essentials of the definition of SP-structures as presented
in~\cite{SP:IJTP}, with minimal explanations. The reader is referred
to~\cite{SP:IJTP} for a gentle complete introduction.
\begin{property}[Symmetry] \label{prop:symmetry}
For any \mbox{$x , y \in \Omega$}, \mbox{$s(y, x) =$} \mbox{$s(x, y)$}.
\end{property}
Symmetry is an experimentally verifiable and fundamental
property of Quantum Mechanics, see,
e.g., the Law of Reciprocity in~\cite{Peres:QuantumTheory}, p. 35.
\begin{property}[Non-negativity] \label{prop:non-negativity}
For any \mbox{$x , y \in \Omega$}, \mbox{$s(x, y) \geq 0$}.
\end{property}
\begin{property}[Boundedness] \label{prop:boundedness}
For any state \mbox{$x \in \Omega$} and any ortho-set $A$,
\mbox{$s(x, A) \eqdef \sum_{a \in A} s(x , a) \leq 1$}.
\end{property}
The next property we want to consider deals with orthogonal projections.
\begin{property}[O-Projection] \label{prop:Projection}
Suppose \mbox{$x \in \Omega$} is a state and \mbox{$A \subseteq \Omega$}
is an ortho-set such that \mbox{$s(x, A) < 1$}.
Then there exists a state \mbox{$y \in \Omega$} with the following properties:
\begin{enumerate}
\item \label{perpA} \mbox{$y \perp A$}, i.e., \mbox{$s(y, A) = 0$}, i.e.,
\mbox{$A \cup \{ y \}$} is an ortho-set, and
\item \label{pxAy} \mbox{$s(x, A) + s(x, y) = 1$}.
\end{enumerate}
\end{property}
O-Projection should remind the reader of the Gram-Schmidt process.
Physically, the ortho-set $A$ represents certain values of a given observable
and therefore can be interpreted as a test: is the state $x$ in $A$ or not.
If \mbox{$s(x, A) <$} $1$ the answer to the question above may, with a certain
``probability'' be ``no''. If the answer is indeed ``no'' the system is left
in a state $y$ that satisfies the three conditions above.
The scalar product
can be seen to satisfy those conditions, when $y$ is the projection of $x$
on the subspace $A^{\perp}$ orthogonal to $A$. In a classical system,
 \mbox{$s(x, A) <$} $1$ implies \mbox{$s(x, A) =$} $0$ and we can take
\mbox{$y =$} $x$.
\begin{definition} \label{def:subspace}
If $A$ is an ortho-set, the subspace \mbox{$\bar{A} \subseteq \Omega$}
generated by $A$ is defined by:
\mbox{$\bar{A} \: = \:$}
\mbox{$\{x \in \Omega \mid s(x, A) = 1 \}$}.
The ortho-set $A$ is said to be a basis for $\bar{A}$.
A {\em basis} is a basis for $\Omega$.
A subspace is a set of states \mbox{$X \subseteq \Omega$}
such that there exists
some ortho-set $A$ such that \mbox{$X = \bar{A}$}.
\end{definition}
Our next defining property for SP-structure is a factorization property.
\begin{property}[Factorization] \label{prop:Factorization}
Let $A$ be an ortho-set and $x$ an arbitrary state.
If \mbox{$y, z \in \bar{A}$} and \mbox{$s(x, y) = s(x, A)$}, then
\mbox{$s(x, z) = s(x, y) \, s(y, z)$}.
\end{property}
Factorization implies that $s(x, A)$ is the maximum of all $s(x, y)$ for
\mbox{$y \in \bar{A}$} and that every such $s(x, y)$ can be factored out
through the state taking this maximum.
Factorization has been described in
Theorem~1 of~\cite{Qsuperp:IJTP}.
The meaning of Factorization, for Physics, is that, if one knows that
in state $y$ some observable $A$ has a specific value,
then the probability of a transition from $x$ to $y$ is the product
of the probability of measuring this specific value (in $x$)
times the transition probability from the state
obtained after the measurement to $y$.
Factorization seems to be a logical requirement relating tests
to two propositions one of which entails the other: if $A$ entails $B$, testing
for $A$ may be done by testing first for $B$ and then for $A$.

In~\cite{SP:IJTP} a last property is presented that is shown to imply
Equation~\ref{eq:cont} of Section~\ref{sec:general}. Since, in this paper,
we only need Equation~\ref{eq:cont}, we shall not present this property here.

\begin{definition} \label{def:equivalence}
Any two states \mbox{$x, y \in \Omega$} are said to be {\em equivalent}, and
we write \mbox{$x \sim y$} iff for any \mbox{$z \in \Omega$}, one has:
\mbox{$s(x, z) = s(y, z)$}. An SP-structure is said to be {\em standard} if
any two equivalent states are equal.
\end{definition}

\bibliographystyle{plain}

\end{document}